\documentclass[12pt,a4paper]{amsart}
\usepackage{a4wide}
\usepackage{amsmath}
\usepackage{amsthm}
\usepackage{amssymb}
\usepackage{amsthm}

\usepackage{graphicx}
\usepackage{color}
\usepackage{caption}
\usepackage{subcaption}
\usepackage{url}

\newtheorem{proposition}{Proposition}
\newtheorem{lemma}[proposition]{Lemma}

\newtheorem{conjecture}[proposition]{Conjecture}

\newtheorem{theorem}[proposition]{Theorem}

\theoremstyle{remark}
\newtheorem{remark}{Remark}
\newtheorem{example}{Example}

\allowdisplaybreaks[3]

\title{On the secrecy gain of $\ell$-modular lattices}
\author{Esa V. Vesalainen and Anne-Maria Ernvall-Hyt\"onen}
\address{Matematik och Statistik, {\AA}bo Akademi University, Domkyrkotorget 1, 20500 {\AA}bo, Finland}
\thanks{This work was supported by the Academy of Finland project 303820, and E.~V.~V. was supported by the Magnus Ehrnrooth Foundation, by the Basque Government through the BERC 2014--2017 program and by Spanish Ministry of Economy and Competitiveness MINECO: BCAM Severo Ochoa excellence accreditation SEV-2013-0323.}
\begin{document}
\maketitle
\begin{abstract}
We show that for every $\ell>1$, there is a counterexample to the $\ell$-modular secrecy function conjecture by Oggier, Sol\'e and Belfiore. These counterexamples all satisfy the modified conjecture by Ernvall-Hyt\"onen and Sethuraman. Furthermore, we provide a method to prove  or disprove the modified conjecture for any given $\ell$-modular lattice rationally equivalent to a suitable amount of copies of $\mathbb{Z}\oplus \sqrt{\ell}\,\mathbb{Z}$ with $\ell \in \{3,5,7,11,23\}$. We also provide a variant of the method for strongly $\ell$-modular lattices when $\ell\in \{6,14,15\}$.
\end{abstract}

\section{Introduction}
Wyner \cite{Wyner} introduced the wiretap channel as a discrete memoryless and possibly noisy broadcast channel, where the sender transmits confidential messages to a legitimate receiver in the presence of an eavesdropper. 
Belfiore and Oggier defined in \cite{belfioggiswire} the secrecy gain
\[
\max_{y\in \mathbb{R_+}}\frac{\Theta_{\mathbb{Z}^n}(y)}{\Theta_{\Lambda}(y)},
\]
where
\[
\Theta_{\Lambda}(y)=\sum_{x\in \Lambda}e^{-\pi\left\|x\right\|^2y}
\]
for $y\in\mathbb R_+$,
as a lattice invariant to measure how much confusion the eavesdropper will experience when a unimodular lattice $\Lambda$ is used in Gaussian wiretap coding. Here we have simplified notation by writing $\Theta_\Lambda(y)$ instead of the traditional $\Theta_\Lambda(yi)$ as this is more convenient. The function $\Xi_{\Lambda}(y)=\Theta_{\mathbb{Z}^n}(y)/\Theta_{\Lambda}(y)$ is called the secrecy function. Belfiore and Sol\'e \cite{belfisole} conjectured that the secrecy function attains its maximum at $y=1$. This function was further studied by Oggier, Sol\'e and Belfiore \cite{belfisoleoggis}, by Lin and Oggier \cite{oglinitw,oglinarxiv}, and by Ernvall-Hyt\"onen and Hollanti \cite{amecami}. A method to prove the conjecture for any given lattice was derived by Ernvall-Hyt\"onen \cite{ErnHyt}, and a new proof for this method was given by Pinchak and Sethuraman \cite{Julie,JulieSeth}.

An $n$-dimensional integral lattice $\Lambda$ is $\ell$-modular, where $\ell\in\mathbb Z_+$, if there exists a similarity $\sigma$ of $\mathbb R^n$ multiplying norms by $\ell$ and mapping $\Lambda^\ast$ to $\Lambda$. Oggier, Sol\'e and Belfiore \cite{belfisoleoggis} defined the secrecy function $\Xi_\Lambda$ for $\ell$-modular lattices by
\[
\Xi_\Lambda(y)=\frac{\Theta_{(\ell^{1/4}\mathbb{Z})^n}(y)}{\Theta_{\Lambda}(y)},
\] 
again for $y\in\mathbb R_+$,
when $n$ is the dimension of the lattice $\Lambda$. So, the lattice $\Lambda$ is compared against a cubic lattice with the same volume. This quantity for $2$-modular and $3$-modular lattices was studied by Lin, Oggier and Sol\'e \cite{LinOggSole}, and for $5$-modular lattices by Hou, Lin and Oggier \cite{houlinoggier}, and in the thesis by Lin \cite{lin:thesis}. Furthermore, it was conjectured that this function obtains its maximum at the natural symmetry point $1/\sqrt{\ell}$. It was proved by Ernvall-Hyt\"onen and Sethuraman that this is not always the case: they provided the $4$-modular lattice $2\mathbb{Z}\oplus \sqrt{2}\,\mathbb{Z}\oplus \mathbb{Z}$ as a counterexample. Strey \cite{strey} provided pictures of some other counterexamples. In this paper, we show that these values of $\ell$ are not exceptions: there are counterexamples for every integer value $\ell>1$.

Ernvall-Hyt\"onen and Sethuraman \cite{e-hsethuraman} suggested normalizing with a suitable power of the theta function of the lattice $D^{\ell}$, instead of the scaled cubic lattice:
\[
\widetilde{\Xi}_{\Lambda}(y)=\frac{\Theta_{D^{\ell}}^{n/2}(y)}{\Theta_{\Lambda}(y)},
\]
where $n$ is the dimension of the lattice $\Lambda$ and $y\in\mathbb R_+$. It is worth noting that $n$ is even unless $\ell$ is a square.
They conjectured that the secrecy function conjecture would be true when the normalization is done using this lattice instead of the cubic lattice. This normalization has the added benefit that the lattice $D^{\ell}$ is $\ell$-modular (not necessarily strongly $\ell$-modular). Furthermore, they gave a method to prove or disprove this modified conjecture for any given $2$-modular lattice.

We will show that whenever $\ell\in \{3,5,7,11,23\}$ and an $\ell$-modular lattice is rationally equivalent to the direct sum of $n/2$ copies of $D^{\ell}$, to check whether the function $\Theta_{\Lambda}$ satisfies the modified conjecture formulated by Ernvall-Hyt\"onen and Sethuraman, it suffices to check whether a certain polynomial obtains its minimal value on the interval
\[\left]0,\left(\frac{\eta\!\left(\frac{1}{2\sqrt{\ell}}\right)\eta\!\left(\frac{2}{\sqrt{\ell}}\right)\eta\!\left(\frac{\sqrt{\ell}}{2}\right)\eta\!\left(2\sqrt{\ell}\right)}{\eta^2\!\left(\frac{1}{\sqrt{\ell}}\right)\eta^2\!\left(\sqrt{\ell}\right)}\right)^{\alpha_\ell}\right]\]
at the right endpoint, where $\eta$ is the classical Dedekind $\eta$-function discussed in more detail later and the exponent $\alpha_\ell$ is defined precisely in Section \ref{polynomization}. %$\frac{\eta\left(\frac{i}{2\sqrt{\ell}}\right)\eta\left(\frac{2i}{\sqrt{\ell}}\right)\eta\left(\frac{i\sqrt{\ell}}{2}\right)\eta\left(2i\sqrt{\ell}\right)}{\eta\left(\frac{i}{\sqrt{\ell}}\right)^2\eta(\sqrt{i\ell})^2}$ when $\ell$ is odd.
The polynomial is obtained from the representation for $\Theta_\Lambda$ given by Rains and Sloane \cite{rainssloane}.

This paper is structured as follows: We first in Section \ref{convolutionidentities} give a useful convolution identity. We will then move in Section \ref{hernandezsethuraman} to the conjecture by Hernandez and Sethuraman about the behavior of the function $\vartheta_3$, and we determine exactly when the conjecture is true. In Section \ref{counterexamples} we give counterexamples to the original $\ell$-modular conjecture and prove that they satisfy the modified conjecture. In Section \ref{polynomization}, we give a method to prove or disprove the modified conjecture under certain conditions for any given $\ell$-modular lattice for certain values of $\ell$. The last two Sections \ref{technical-theta-section} and \ref{technical-eta-section} give some rather technical lemmas about $\vartheta_3$ and $\eta$ needed in the proofs.

\section{A convolution identity}\label{convolutionidentities}

The following simple convolution identity will be quite useful for us later. For any given real numbers $h$ and $k$ with $0\leqslant k<h$, we define the auxiliary trapezoid function $T(\cdot;k,h)\colon\mathbb R\longrightarrow\mathbb R$, depicted in Figure \ref{figure-trapezoid}, by setting
\[T(x;k,h)=\begin{cases}
0&\text{when $x\leqslant-h$,}\\
x+h&\text{when $-h\leqslant x\leqslant-k$,}\\
h-k&\text{when $-k\leqslant x\leqslant k$,}\\
h-x&\text{when $k\leqslant x\leqslant h$, and}\\
0&\text{when $x\geqslant h$,}
\end{cases}\]
for all $x\in\mathbb R$. Of course, for $k=0$ the trapezoid reduces to a triangle.
\begin{figure}[h]
\begin{center}
\includegraphics[scale=.8]{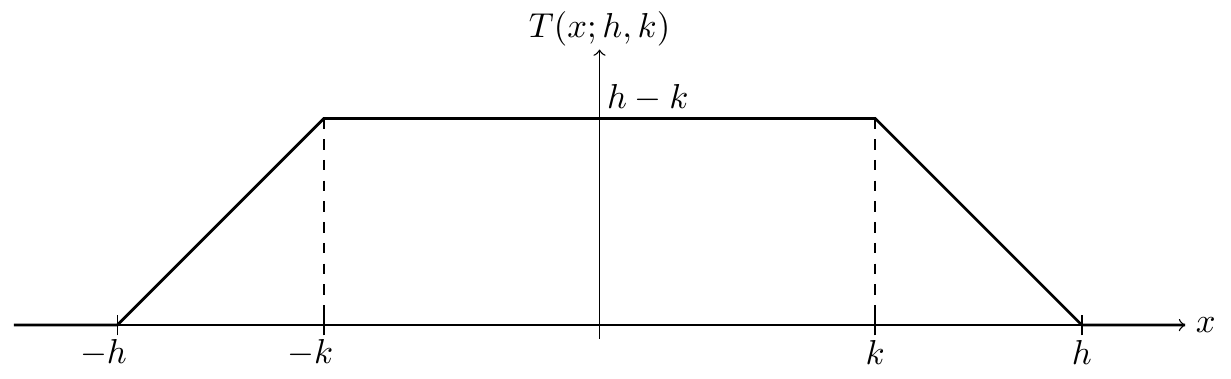}
\end{center}
\caption{\label{figure-trapezoid} The function $T(x;k,h)$.}
\end{figure}

\begin{lemma}\label{convolution-identity}
Let $f\colon\mathbb R\longrightarrow\mathbb R$ be a twice continuously differentiable function, and let $h,k\in\left[0,\infty\right[$ with $k<h$. Then
\[f(x+h)-f(x+k)-f(x-k)+f(x-h)
=\int\limits_{-\infty}^\infty f''(u)\,T(u-x;k,h)\,\mathrm du,\]
for all $x\in\mathbb R$.
\end{lemma}

\begin{proof}
By straightforward calculation, we have
\begin{align*}
&\int\limits_{-\infty}^\infty f''(u)\,T(u-x;k,h)\,\mathrm du
=\int\limits_{-\infty}^\infty f''(x+u)\,T(u;k,h)\,\mathrm du\\
&=\int\limits_{-h}^{-k}f''(x+u)\,(u+h)\,\mathrm du
+\int\limits_{-k}^kf''(x+u)\,(h-k)\,\mathrm du
+\int\limits_k^hf''(x+u)\,(h-u)\,\mathrm du\\
&=\left.f'(x+u)\,(u+h)\vphantom{\Big|}\right]_{-h}^{u=-k}
-\int\limits_{-h}^{-k}f'(x+u)\,\mathrm du
+(h-k)\,f'(x+k)-(h-k)\,f'(x-k)\\&\qquad
+\left.f'(x+u)\,(h-u)\vphantom{\Big|}\right]_k^{u=h}
+\int\limits_k^hf'(x+u)\,\mathrm du\\
&=(h-k)\,f'(x-k)-f(x-k)+f(x-h)+(h-k)\,f'(x+k)\\&\qquad-(h-k)\,f'(x-k)-(h-k)\,f'(x+k)+f(x+h)-f(x+k)\\
&=f(x+h)-f(x+k)-f(x-k)+f(x-h),
\end{align*}
as claimed.
\end{proof}

\section{Conjecture of Hernandez and Sethuraman}\label{hernandezsethuraman}

The classical function $\vartheta_3$ is defined by setting
\[\vartheta_3(y)=\Theta_{\mathbb Z}(y)=\sum_{n\in\mathbb Z}e^{-\pi n^2y}=1+2\sum_{n=1}^\infty e^{-\pi n^2y}\]
for all $y\in\mathbb R_+$, where we have employed the usual simplification of notation by writing $\vartheta_3(y)$ instead of $\vartheta_3(yi)$. The results in Section \ref{counterexamples} depend on understanding the behaviour of certain kinds of expressions involving $\vartheta_3$.

Hernandez and Sethuraman \cite{hernandez,Sethuraman} made the following conjecture which is Conjecture 8 in \cite{hernandez}:
\begin{conjecture}[Hernandez and Sethuraman]\label{hernandez--sethuraman-conjecture} Let $a,b\in\mathbb R_+$. Then the expression
\[
\frac{\vartheta_3(y)\,\vartheta_3(aby)}{\vartheta_3(ay)\,\vartheta_3(by)},
\]
defined for $y\in\mathbb R_+$,
obtains its unique maximum at ${1}/{\sqrt{ab}}$.
\end{conjecture}
\noindent
This conjecture was supported by several illustrations. The following theorem implies that the conjecture is essentially true, except for some ranges of $a$ and $b$ where the behaviour of the expression is naturally opposite. In Section \ref{counterexamples}, Theorem \ref{theta-quotients} is used to both derive counterexamples to the original $\ell$-modular conjecture and to prove that the counterexamples satisfy the modified conjecture.

\begin{theorem}\label{theta-quotients}
Let $\kappa,\lambda\in\mathbb R_+$ with $1\leqslant\kappa<\lambda$, and let us define a function $g\colon\mathbb R_+\longrightarrow\mathbb R_+$ by setting
\[g(y)=\frac{\vartheta_3(\lambda y)\,\vartheta_3\!\left(y/\lambda\right)}{\vartheta_3(\kappa y)\,\vartheta_3(y/\kappa)}\]
for $y\in\mathbb R_+$. Then the function $g$ has a strict global maximum at the point $1$, is strictly increasing in $\left]0,1\right]$, and is strictly decreasing in $\left[1,\infty\right[$.
\end{theorem}

\begin{figure}[h]
\begin{center}
\includegraphics[scale=.6]{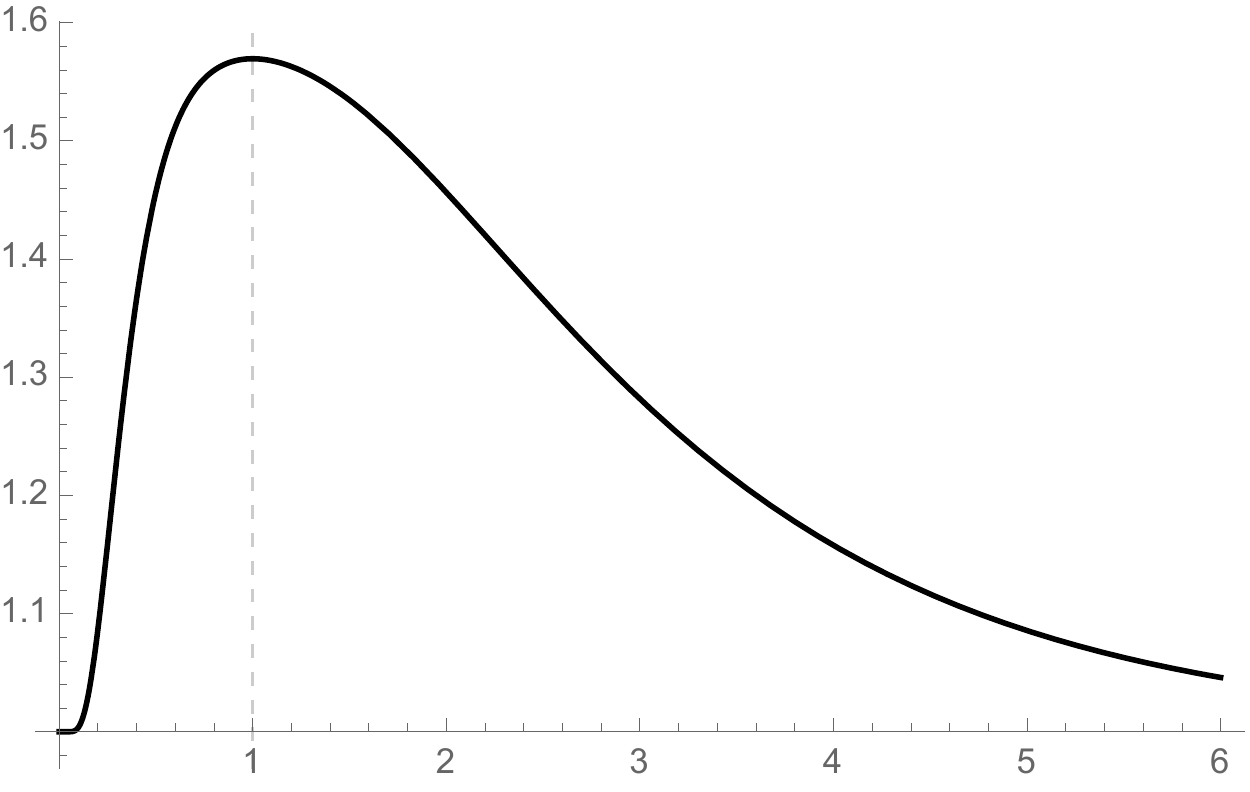}
\end{center}
\caption{\label{figure-theta-quotient} The function $g(y)$ of Theorem \ref{theta-quotients} with $\kappa=2$ and $\lambda=5$.}
\end{figure}

The proof of the theorem requires some very technical lemmas which are stated and proved in Section \ref{technical-theta-section}.

\begin{proof}
Let us first observe that by the modularity relation satisfied by $\vartheta_3$, we have
\[\vartheta_3\!\left(\frac1y\right)=\sqrt y\,\vartheta_3(y)\]
for all $y\in\mathbb R_+$, and that this implies easily that
\[g\!\left(\frac1y\right)=g(y)\]
for all $y\in\mathbb R_+$. The function $g$ is also clearly real-analytic. Thus, it is enough to prove that $g$ is strictly decreasing in $\left]1,\infty\right[$.

Let us write $f(x)=\log\vartheta_3(e^x)$, $x=\log y$, $h=\log\lambda$ and $k=\log\kappa$. Then
\[\log g(y)=\log\frac{\vartheta_3(\lambda y)\,\vartheta_3\!\left(y/\lambda\right)}{\vartheta_3(\kappa y)\,\vartheta_3(y/\kappa)}
=f(x+h)-f(x+k)-f(x-k)+f(x-h),\]
and we need to prove that the last expression is strictly decreasing for $x\in\mathbb R_+$. By Lemma \ref{convolution-identity}, we may rewrite this expression as a convolution against $T(\cdot;k,h)$ as
\[f(x+h)-f(x+k)-f(x-k)+f(x-h)
=\int\limits_{-\infty}^\infty f''(u)\,T(u-x;k,h)\,\mathrm du.\]

Since $f''$ is infinitely smooth and $T(\cdot;k,h)$ is continuous and compactly supported, the last integral is differentiable with derivative
\[\frac{\mathrm d}{\mathrm dx}\int\limits_{-\infty}^\infty f''(u)\,T(u-x;k,h)\,\mathrm du
=\int\limits_{-\infty}^\infty f'''(u)\,T(u-x;k,h)\,\mathrm du,\]
and the problem is reduced to proving that the last integral is strictly negative for $x\in\mathbb R_+$. Before embarking on this, let us invoke Lemma \ref{theta-min-technical-lemma} which says that the third derivative $f'''$ is an odd function on $\mathbb R$ and strictly negative in $\mathbb R_+$.

Let us now start treating the integral by observing that the function $T(\cdot-x;k,h)$ is supported on $\left[x-h,x+h\right]$, and so
\[\int\limits_{-\infty}^\infty f'''(u)\,T(u-x;k,h)\,\mathrm du
=\int\limits_{x-h}^{x+h}f'''(u)\,T(u-x;k,h)\,\mathrm du.\]
In particular, if $x\geqslant h$, then the integrand is strictly negative in $\left]x-h,x+h\right[$. Thus, we may assume that $x\in\left]0,h\right[$.

We now split the integral into pieces and rearrange them, remembering that $f'''(\cdot)$ is an odd function, and that $T(\cdot;k,h)$ is an even function, to get
\begin{align*}
&\int\limits_{x-h}^{x+h}f'''(u)\,T(u-x;k,h)\,\mathrm du\\
&=\int\limits_{x-h}^0f'''(u)\,T(u-x;k,h)\,\mathrm du
+\int\limits_0^{h-x}\,f'''(u)\,T(u-x;k,h)\,\mathrm du
+\int\limits_{h-x}^{h+x}f'''(u)\,T(u-x;k,h)\,\mathrm du\\
&=-\int\limits_0^{h-x}f'''(u)\,T(u+x;k,h)\,\mathrm du
+\int\limits_0^{h-x}f'''(u)\,T(u-x;k,h)\,\mathrm du
+\int\limits_{h-x}^{h+x}f'''(u)\,T(u-x;k,h)\,\mathrm du\\
&=\int\limits_0^{h-x}f'''(u)\left(T(u-x;k,h)-T(u+x;k,h)\right)\mathrm du
+\int\limits_{h-x}^{h+x}f'''(u)\,T(u-x;k,h)\,\mathrm du.
\end{align*}
Here the last integral is again strictly negative since the integrand is, and so it is enought to prove that the second to the last integral is nonnegative. Since $f'''$ is again strictly negative in the integrand in $\left]0,h-x\right[$, it is enough to prove that $T(u+x;k,h)\leqslant T(u-x;k,h)$ for $u\in\left]0,h-x\right[$. But this last inequality is geometrically obvious, and we are done.
\end{proof}

Let us now look at the original conjecture. Assume that $a,b\in\mathbb R_+$ with $a\leqslant b$. If $a=1$ or $b=1$ then the $\vartheta_3$-expression reduces to a constant, so let us assume that $a\neq1$ and $b\neq1$. We can write the expression as
\[
\frac{\vartheta_3(y)\,\vartheta_3(aby)}{\vartheta_3(ay)\,\vartheta_3(by)}
=\frac{\vartheta_3\!\left(\frac{1}{\sqrt{ab}}\,(\sqrt{ab}\,y)\right)\vartheta_3\!\left(\sqrt{ab}\,(\sqrt{ab}\,y)\right)}
{\vartheta_3\!\left(\sqrt{\frac{\vphantom ba}{b}}\,(\sqrt{ab}\,y)\right)\vartheta_3\!\left(\sqrt{\frac{b}{\vphantom ba}}\,(\sqrt{ab}\,y)\right)}.
\]
Clearly $\sqrt{b/a}\geqslant 1\geqslant \sqrt{a/b}$. If $1<a\leqslant b$, then
\[\sqrt{ab}>\sqrt{\frac{b}{\vphantom ba}}\geqslant \sqrt{\frac{\vphantom ba}{b}}>\frac{1}{\sqrt{ab}},\]
and Conjecture \ref{hernandez--sethuraman-conjecture} holds by Theorem \ref{theta-quotients}. Similarly, if $a\leqslant b<1$, then
\[
\frac{1}{\sqrt{ab}}>\sqrt{\frac{b}{\vphantom ba}}\geqslant \sqrt{\frac{\vphantom ba}{b}}>\frac{1}{\sqrt{ab}},
\]
and again Conjecture \ref{hernandez--sethuraman-conjecture} holds by Theorem \ref{theta-quotients}.
Finally, if $a<1<b$, then
\[
\sqrt{\frac{b}{\vphantom ba}}>\sqrt{ab}>\sqrt{\frac{\vphantom ba}b},
\quad\text{and}\quad
\sqrt{\frac b{\vphantom ba}}>\frac{1}{\sqrt{ab}}>\sqrt{\frac{\vphantom ba}{b}},
\]
and the conjecture is not true, because, by Theorem \ref{theta-quotients}, the function obtains a minimum instead of a maximum at the points $1/\sqrt{ab}$.

\section{Counterexamples}\label{counterexamples}

We would like to start this section by giving very elementary counterexamples for the original conjecture. Let $\ell\in\mathbb Z_+$ and $\ell\geqslant2$, and let us consider the $\ell$-modular lattice
\[D^{\ell}=\mathbb Z\oplus\sqrt\ell\,\mathbb Z.\]
The secrecy function attached to this lattice is
\[\Xi_{D^{\ell}}(y)=\frac{\Theta_{(\ell^{1/4}\mathbb Z)^2}(y)}{\Theta_{D^{\ell}}(y)}=\frac{\vartheta_3^2(y\sqrt\ell)}{\vartheta_3(y)\,\vartheta_3(y\ell)},\]
defined for $y\in\mathbb R_+$. The secrecy function conjecture states that the function $\Xi_{D^{\ell}}$ should have a global maximum at the natural symmetry point $1/\sqrt\ell$. Our goal here is to prove that this is not the case.

\begin{theorem}\label{theta-ell-corollary}
Let $\ell\in\mathbb Z_+$ and $\ell\geqslant2$. Then the lattice $D^{\ell}$ does not satisfy the secrecy function conjecture. More precisely, the secrecy function $\Xi_{D^{\ell}}$ does not have a global maximum at the point $1/\sqrt\ell$.
\end{theorem}

\begin{proof}
On the one hand,
\[\Xi_{D^{\ell}}(y)\longrightarrow1\quad\text{as}\quad y\longrightarrow\infty.\] On the other hand, Lemma \ref{theta-ell-lemma} below tells us that \[\Xi_{D^{\ell}}\!\left(\frac1{\sqrt\ell}\right)<1.\]
Thus, clearly the secrecy function cannot have a global maximum at the point $1/\sqrt\ell$.
\end{proof}

\begin{lemma}\label{theta-ell-lemma}
Let $\ell\in\mathbb Z_+$ and $\ell\geqslant2$. Then we have
\[\Xi_{D^{\ell}}\!\left(\frac1{\sqrt\ell}\right)<1.\]
\end{lemma}

\begin{proof}
We need to prove that
\[\vartheta_3^2(1)<\vartheta_3\!\left(\frac1{\sqrt\ell}\right)\vartheta_3(\sqrt\ell).\]
We recall that the $\vartheta$-function $\vartheta_3$ satisfies the modularity relation
\[\vartheta_3\!\left(\frac1y\right)=\sqrt y\,\vartheta_3(y)\]
for all $y\in\mathbb R_+$. Given this modularity relation, and the fact that $\vartheta_3(y)>1$ for all $y\in\mathbb R_+$, we have
\[\vartheta_3\!\left(\frac1{\sqrt\ell}\right)\vartheta_3(\sqrt\ell)
=\ell^{1/4}\,\vartheta_3^2(\sqrt\ell)>\ell^{1/4}.\]
On the other hand, it is easy to compute numerically that
\[\vartheta_3^8(1)\approx1.9410\ldots<2\leqslant\ell,\]
and so $\vartheta_3^2(1)<\ell^{1/4}$, and we are done.
\end{proof}

We can use the more advanced Theorem  \ref{theta-quotients} to prove the following counterexamples. The purpose of the condition $a_1<\sqrt\ell$  in the following theorem is just to exclude the uninteresting case $L=(\sqrt\ell\,\mathbb Z)^{n+1}$ in which the secrecy function $\Xi_L$ is a constant function. The requirement that $1<a_k<\ell$ for some $k$ in Theorem \ref{modified-theta-ell-max-general} serves the same purpose.

\begin{theorem}\label{theta-ell-min-general}
Let $\ell,n\in\mathbb Z_+$ with $\ell\geqslant2$, and let us be given integers
\[1\leqslant a_1\leqslant a_2\leqslant a_3\leqslant\ldots\leqslant a_n\leqslant\ell,\]
such that $a_1<\sqrt\ell$ and $a_ka_{n+1-k}=\ell$ for each $k\in\left\{1,2,\ldots,n\right\}$. Then the secrecy function $\Xi_L$ of the $\ell$-modular lattice
\[L=\bigoplus_{k=1}^n\sqrt{a_k}\,\mathbb Z\]
has a unique global minimum at the point $1/\sqrt\ell$, is strictly decreasing in $\bigl]0,1/\sqrt\ell\,\bigr]$ and strictly increasing in $\bigl[1/\sqrt\ell,\infty\bigr[$.
\end{theorem}

\begin{proof}
Since $\Xi_L$ only takes positive values, we may consider its square which may be written in the form
\[\Xi_L^2\!\left(\frac y{\sqrt\ell}\right)=
\prod_{k=1}^n\frac{\vartheta_3^2(y)}{\vartheta_3(a_ky/\sqrt\ell)\,\vartheta_3(a_{n+1-k}y/\sqrt\ell)},\]
and the desired properties of $\Xi_L$ follow from Theorem \ref{theta-quotients}.
\end{proof}

\begin{remark}
In \cite{Faulhuber--Steinerberger} Faulhuber and Steinerberger proved that the function $\log\vartheta_3(e^x)$ is strictly convex for $x\in\mathbb R$, which is actually enough to see that the above lattices $L$ are counterexamples, as Jensen's inequality directly gives the weaker conclusion $\Xi_L(y)<1$ for all $y\in\mathbb R_+$.
\end{remark}

\begin{theorem}\label{modified-theta-ell-max-general}
Let $\ell,n\in\mathbb Z_+$ with $\ell\geqslant2$, and let us be given integers
\[1\leqslant a_1\leqslant a_2\leqslant a_3\leqslant\ldots\leqslant a_n\leqslant\ell,\]
such that and $a_ka_{n+1-k}=\ell$ for each $k\in\left\{1,2,\ldots,n\right\}$ and that $1<a_k<\ell$ for at least one $k\in\left\{1,2,\ldots,n\right\}$. Then the modified secrecy function $\widetilde\Xi_L$ of the $\ell$-modular lattice
\[L=\bigoplus_{k=1}^n\sqrt{a_k}\,\mathbb Z\]
has a unique global maximum at the point $1/\sqrt\ell$, is strictly increasing in $\bigl]0,1/\sqrt\ell\,\bigr]$, and strictly decreasing in $\bigl[1/\sqrt\ell,\infty\bigr[$.
\end{theorem}

\begin{proof}
Since $\widetilde\Xi_L$ takes only positive values, we may consider its square which may be written in the form
\[\widetilde\Xi_L^2\!\left(\frac y{\sqrt\ell}\right)=
\prod_{k=1}^n\frac{\vartheta_3(y\sqrt\ell)\,\vartheta_3(y/\sqrt\ell)}{\vartheta_3(a_ky/\sqrt\ell)\,\vartheta_3(a_{n+1-k}y/\sqrt\ell)},\]
and the desired properties of $\widetilde\Xi_L$ follow from Theorem \ref{theta-quotients}.
\end{proof}

\begin{example}
For example, let us consider the lattice
\[C^\ell=\bigoplus_{d\mid n}\sqrt d\,\mathbb Z,\]
where $d\mid n$ means taking a term for each positive divisor $d$ of $\ell$, and where $\ell\in\mathbb Z_+$ with $\ell\geqslant2$. By Theorem \ref{theta-ell-min-general} this lattice is a counterexample to the original $\ell$-modular secrecy function conjecture, but by Theorem \ref{modified-theta-ell-max-general} it does satisfy the modified $\ell$-modular secrecy function conjecture.
\end{example}

\begin{figure}[h]
\begin{center}
\includegraphics[scale=.6]{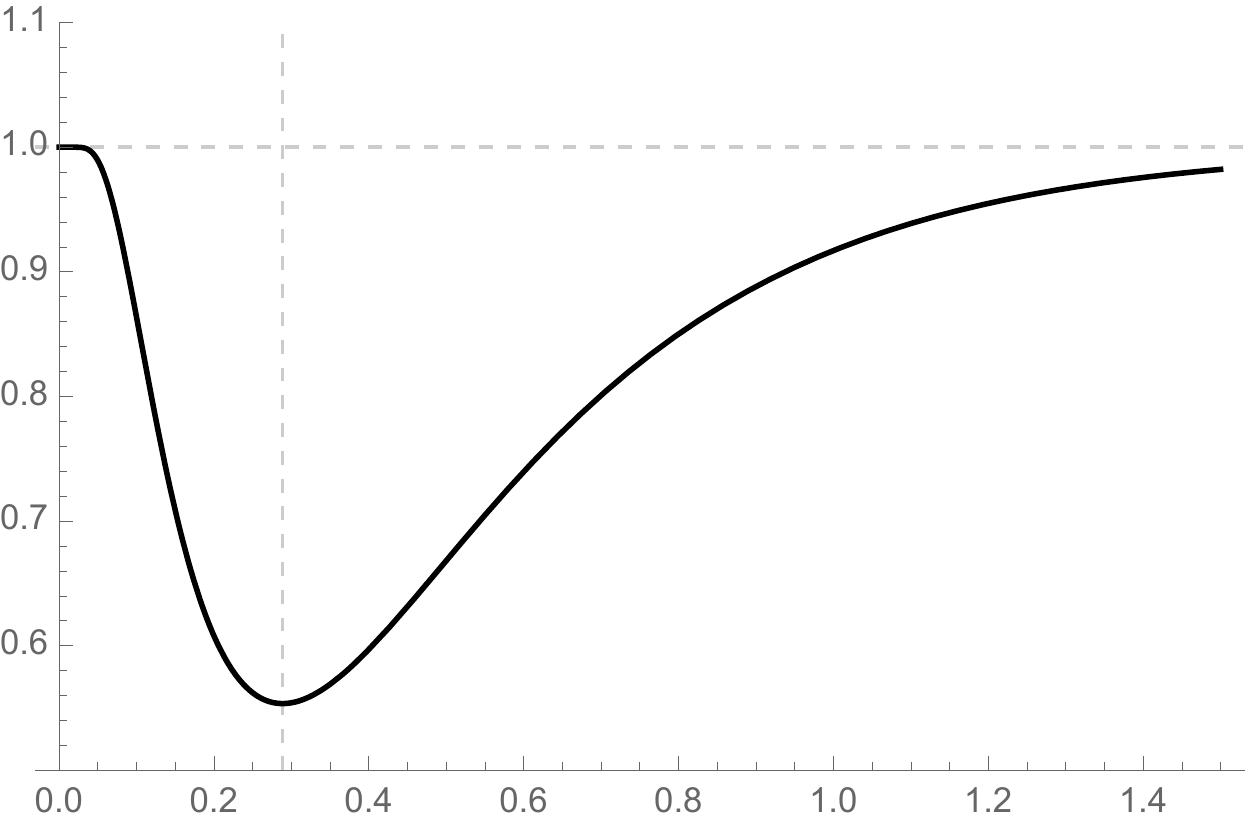}
\quad
\includegraphics[scale=.6]{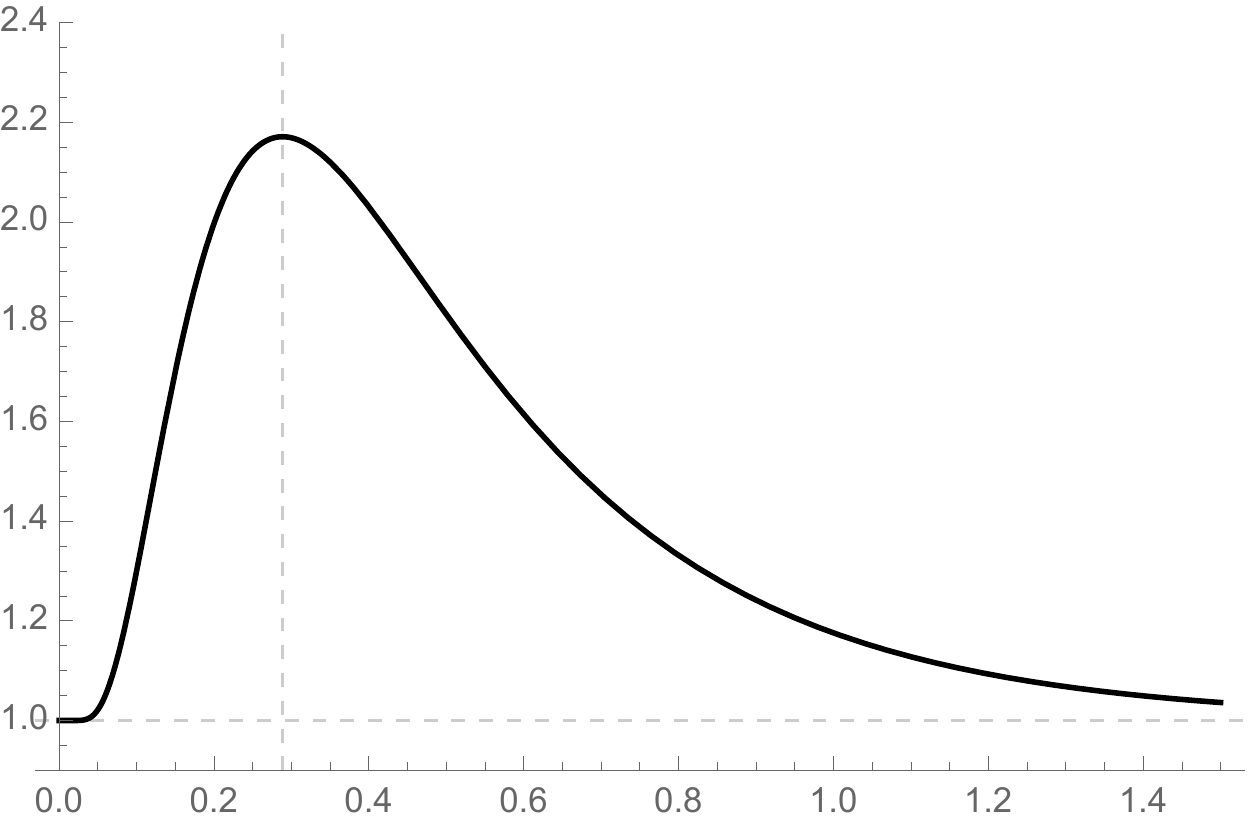}
\end{center}
\caption{\label{figure-cell} The secrecy function $\Xi_{C^{12}}$ and the modified secrecy function $\widetilde\Xi_{C^{12}}$.}
\end{figure}

\section{Polynomization}\label{polynomization}

The classical Dedekind $\eta$-function is defined by setting
\[\eta(y)=e^{-\pi y/12}\prod_{n=1}^\infty\left(1-e^{-2\pi ny}\right)\]
for $y\in\mathbb R_+$, where we have again simplified notation by writing $\eta(y)$ for $\eta(yi)$, as usual. We recall that this satisfies the modularity relation
\[\eta\!\left(\frac1y\right)=\sqrt y\,\eta(y)\]
for all $y\in\mathbb R_+$.

Let $\ell\in\mathbb Z_+$. For the polynomization we need to understand $\eta$-quotients involving the function $\eta^{(\ell)}\colon\mathbb R_+\longrightarrow\mathbb R_+$ defined by
\[\eta^{(\ell)}(y)=\prod_{d\mid\ell}\eta(dy)\]
for all $y\in\mathbb R_+$, where the product $\prod_{d\mid\ell}$ is taken over all positive divisors $d$ of $\ell$. When $\ell$ is odd, we also define
\[
g_\ell(y)=\left(\frac{\eta^{(\ell)}(y/2)\,\eta^{(\ell)}(2y)}{\left(\eta^{(\ell)}(y)\right)^2}\right)^{D_\ell/\dim C^\ell}
\]
for $y\in\mathbb R_+$,
and when $\ell$ is even, we define
\[
g_\ell(y)=\left(\frac{\eta^{(\ell/2)}(y/2)\,\eta^{(\ell/2)}(4y)}{\eta^{(\ell/2)}(y)\,\eta^{(\ell/2)}(2y)}\right)^{D_\ell/\dim C^\ell},
\]
again for $y\in\mathbb R_+$.
Here
\[D_\ell=\frac{24\,d(\ell)}{\prod_{p\mid\ell}\left(p+1\right)},\]
where $d(\ell)$ is the number of positive divisors of $\ell$ and the product $\prod_{p\mid\ell}$ is over the positive prime divisors $p$ of $\ell$. Naturally, we have
\[\dim C^\ell=\dim\bigoplus_{d\mid n}\sqrt d\,\mathbb Z=d(\ell).\]
%Finally, we simplify notation by letting $O_\ell$ be the order of vanishing of $\Theta_{C^\ell}$, understood as a modular form, at the cusp~$1$.

The following theorem is the first half of Corollary 3 from \cite{rainssloane}. We recall that two lattices $\Lambda\subset\mathbb R^n$ and $\Lambda'\subset\mathbb R^n$ of dimension $n\in\mathbb Z_+$ are called rationally equivalent if $\Lambda=A\,\Lambda'$ for some invertible matrix $A\in\mathbb Q^{n\times n}$.

We also recall that an integral lattice $\Lambda$ is strongly $\ell$-modular for $\ell\in\mathbb Z_+$ if the smallest $\ell'\in\mathbb Z_+$, for which $\sqrt{\ell'}\,\Lambda^\ast$ is integral, satisfies $\ell'\mid\ell$, and if $\Lambda$ has a modularity of level $m$ for each positive divisor $m\mid\ell$ such that $m$ and $\ell/m$ are coprime. A modularity $\sigma$ of level $m\in\mathbb Z_+$ of an $n$-dimensional integral lattice $\Lambda$ is a similarity of $\mathbb R^n$ multiplying norms by $m$ and mapping $\Lambda^{\ast\Pi}$ to $\Lambda$ where $\Pi$ is the set of primes dividing $m$. Finally, given a set of primes $\Pi$ the $\Pi$-dual $\Lambda^{\ast\Pi}$ of an integral lattice $\Lambda$ consists of those vectors $v\in\Lambda\otimes\mathbb Q$ such that $v\cdot\Lambda\subseteq\mathbb Z_p$ for $p\in\Pi$ and $v\cdot\Lambda^\ast\subseteq\mathbb Z_p$ for primes $p\not\in\Pi$. In particular, if $\ell$ is a prime, then an $\ell$-modular integral lattice is also strongly $\ell$-modular.

\begin{theorem} Assume that $\ell\in \{1,2,3,5,6,7,11,14,15,23\}$,
and let $\Lambda$ be a strongly $\ell$-modular lattice that is rationally equivalent to $(C^{\ell})^k$, where $k\in\mathbb Z_+$. Then its theta series can be written in the form
\[
\Theta_\Lambda=\Theta_{C^\ell}^k\sum_{i=0}^{N}c_i\,g_2^i,
\]
where $N\in\mathbb Z_+$ and the coefficients $c_0$, $c_1$, \dots, $c_N$ are real numbers.
\end{theorem}

\begin{remark}
We remark here that one may use the beautiful theory of rational quadratic forms (see e.g.\ Chapter IV in \cite{Serre}) to check the rational equivalence, similarly to how the $2$-modular case was considered in \cite{e-hsethuraman}. More precisely, if we consider the lattice $(C^\ell)^k$ and a given $\ell$-modular lattice $\Lambda$ of the same dimension, where $\ell,k\in\mathbb Z_+$, then the two lattices are rationally equivalent if and only if the corresponding quadratic forms have the same discriminants, signatures and Hasse--Witt invariants. Since they have the same determinants, namely $\ell^k$, the discriminants must also be the same, and since the quadratic forms are strictly positive-definite, this easily implies that also the signatures are the same and that the Hasse--Witt invariants at $\infty$ are both equal to $+1$. The argument used in the appendix of \cite{e-hsethuraman} shows that the Hasse--Witt invariants at $p$ are both equal to $+1$ for any odd prime $p$ not dividing $\ell$. This only leaves finitely many Hasse--Witt invariants $\varepsilon_p$ to compute and to compare. More detailed discussions can be found in \cite{Serre, e-hsethuraman}.
\end{remark}

We can now state the theorems that we are going to prove.

\begin{theorem}\label{poly} Let $\ell\in \{3,5,7,11,23\}$, let $\Lambda$ be an $\ell$-modular lattice, which is rationally equivalent to $(C^{\ell})^k$, where $k\in\mathbb Z_+$, and let
\[
\Theta_{\Lambda}=\Theta_{C^\ell}^k\sum_{i=0}^{N}c_i\,g_\ell^i,
\]
where $N\in\mathbb Z_+$ and the coefficients $c_0$, $c_1$, \dots, $c_N$ are real numbers. Then the modified secrecy function $\widetilde\Xi_\Lambda$ has a unique global maximum at the natural symmetry point $1/\sqrt\ell$ if and only if
\[
\sum_{i=0}^{N}c_i\,x^i>\sum_{i=0}^Nc_i\,g_\ell^i\!\left(\frac1{\sqrt\ell}\right)
\]
for all $x\in\left]0,g_\ell(1/\smash{\sqrt\ell)}\right[$.
%obtains its minimum on the interval $\left[0,\frac{\eta\left(\frac{i}{2\sqrt{\ell}}\right)\eta\left(\frac{2i}{\sqrt{\ell}}\right)\eta\left(\frac{i\sqrt{\ell}}{2}\right)\eta\left(2i\sqrt{\ell}\right)}{\eta\left(\frac{i}{\sqrt{\ell}}\right)^2\eta(\sqrt{i\ell})^2}\right]$ at the point $\frac{\eta\left(\frac{i}{2\sqrt{\ell}}\right)\eta\left(\frac{2i}{\sqrt{\ell}}\right)\eta\left(\frac{i\sqrt{\ell}}{2}\right)\eta\left(2i\sqrt{\ell}\right)}{\eta\left(\frac{i}{\sqrt{\ell}}\right)^2\eta(\sqrt{i\ell})^2}$ when $\ell$ is odd.
\end{theorem}

\begin{proof} Since $C^\ell=\mathbb Z\oplus\sqrt\ell\,\mathbb Z=D^\ell$ as $\ell$ is a prime, we have
\[
\frac1{\widetilde\Xi_{\Lambda}}=\frac{\Theta_{\Lambda}}{\Theta_{D^\ell}^k(y)}=\sum_{i=0}^{N}c_i\, g_\ell^i.
\]
Hence, it suffices to consider the expression $\sum_{i=0}^{N}c_i\,g_\ell^i$. Let us now analyse the behaviour of the function $g_\ell$. By Theorem \ref{eta-quotient-theorem} below the function $g_\ell$ is strictly increasing in $\left]0,1/\smash{\sqrt\ell}\right]$, and strictly decreasing in $\left[1/\smash{\sqrt\ell},\infty\right[$, with $g_\ell(y)\longrightarrow0+$ as $y\longrightarrow0+$ or $y\longrightarrow\infty$. Thus, the conjecture holds for $\Lambda$ if and only if the polynomial $\sum_{i=0}^{N}c_i\,x^i$ has in the interval $\left]0,g_\ell(1/\smash{\sqrt\ell})\right]$ a strict global minimum at the point $g_\ell(1/\sqrt\ell)$, and we are done.
\end{proof}

\begin{example}\label{K12example}
The $\vartheta$-function of the strongly $3$-modular $12$-dimensional Coxeter--Todd lattice $K_{12}$ is of the form
\[
\Theta_{K_{12}}(y)=1+756\,e^{-4\pi y}+4032\,e^{-6\pi y}+20412\,e^{-8\pi y}+60480 \,e^{-10\pi y}+\ldots
\]
It can be written in the form
\[
\left(\vartheta_3(y)\,\vartheta_3(3y)\right)^6\left(1-12\,g_3(y)+12\,g_3^2(y)-64\,g_3^3(y)\right).
\]
The modified secrecy function of this lattice is thus
\begin{alignat*}{1}\widetilde{\Xi}_{K_{12}}(y)&=\frac{\left(\vartheta_3(y)\,\vartheta_3(3y)\right)^6}{\left(\vartheta_3(y)\,\vartheta_3(3y)\right)^6\left(1-12\,g_3(y)+12\,g_3^2(y)-64\,g_3^3(y)\right)}\\&=\left(1-12\,g_3(y)+12\,g_3^2(y)-64\,g_3^3(y)\right)^{-1}\end{alignat*}
Let us now look at the polynomial $P(x)=1-12x+12x^2-64x^3$. Since $P'(x)=-12+24x-192x^2<0$ on real numbers, the polynomial $P(x)$ is decreasing. Hence, the lattice $K_{12}$ satisfies the modified conjecture.
\end{example}

\begin{example}\label{H4example}
The $\vartheta$-function of the strongly $5$-modular $8$-dimensional icosean lattice $H4$, also known as $Q_8(1)$, is of the form
\[
\Theta_{H4}(y)=1+120\,e^{-4\pi y}+240\,e^{-6\pi y}+600\,e^{-8\pi y}+1440 \,e^{-10\pi y}+\ldots
\]
It can be written in the form
\[
\left(\vartheta_3(y)\,\vartheta_3(5y)\right)^4\left(1-8\,g_5(y)+8\,g_5^2(y)-16\,g_5^3(y)\right).
\]
The modified secrecy function of this lattice is thus
\begin{alignat*}{1}\widetilde{\Xi}_{H4}(y)&=\frac{\left(\vartheta_3(y)\,\vartheta_3(5y)\right)^4}{\left(\vartheta_3(y)\,\vartheta_3(5y)\right)^4\left(1-8\,g_5(y)+8\,g_5^2(y)-16\,g_5^3(y)\right)}\\&=\left(1-8\,g_5(y)+8\,g_5^2(y)-16\,g_5^3(y)\right)^{-1}\end{alignat*}
Let us now look at the polynomial $P(x)=1-8x+8x^2-16x^3$. Since $P'(x)=-8+16x-48x^2<0$ on real numbers, the polynomial $P(x)$ is decreasing. Hence, the lattice $H4$ satisfies the modified conjecture.
\end{example}

\begin{figure}[h]
\begin{center}
\includegraphics[scale=.6]{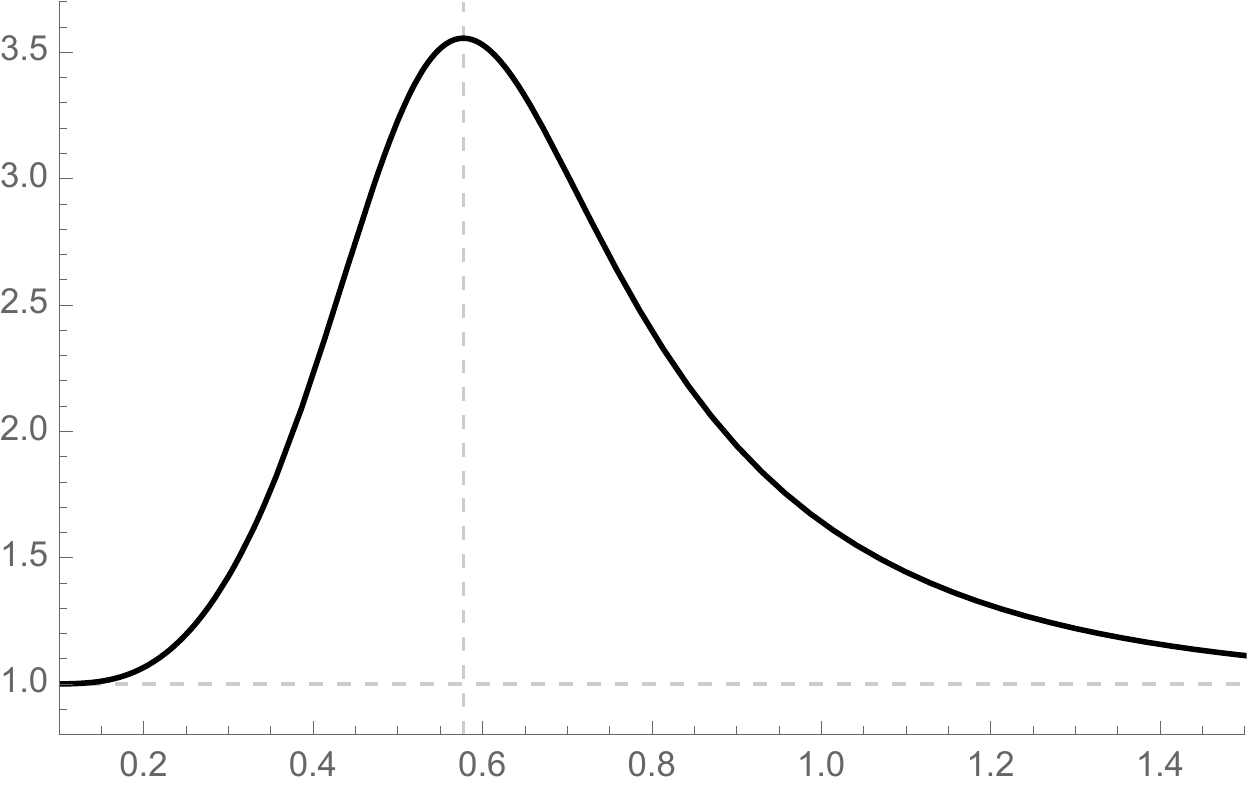}
\quad
\includegraphics[scale=.6]{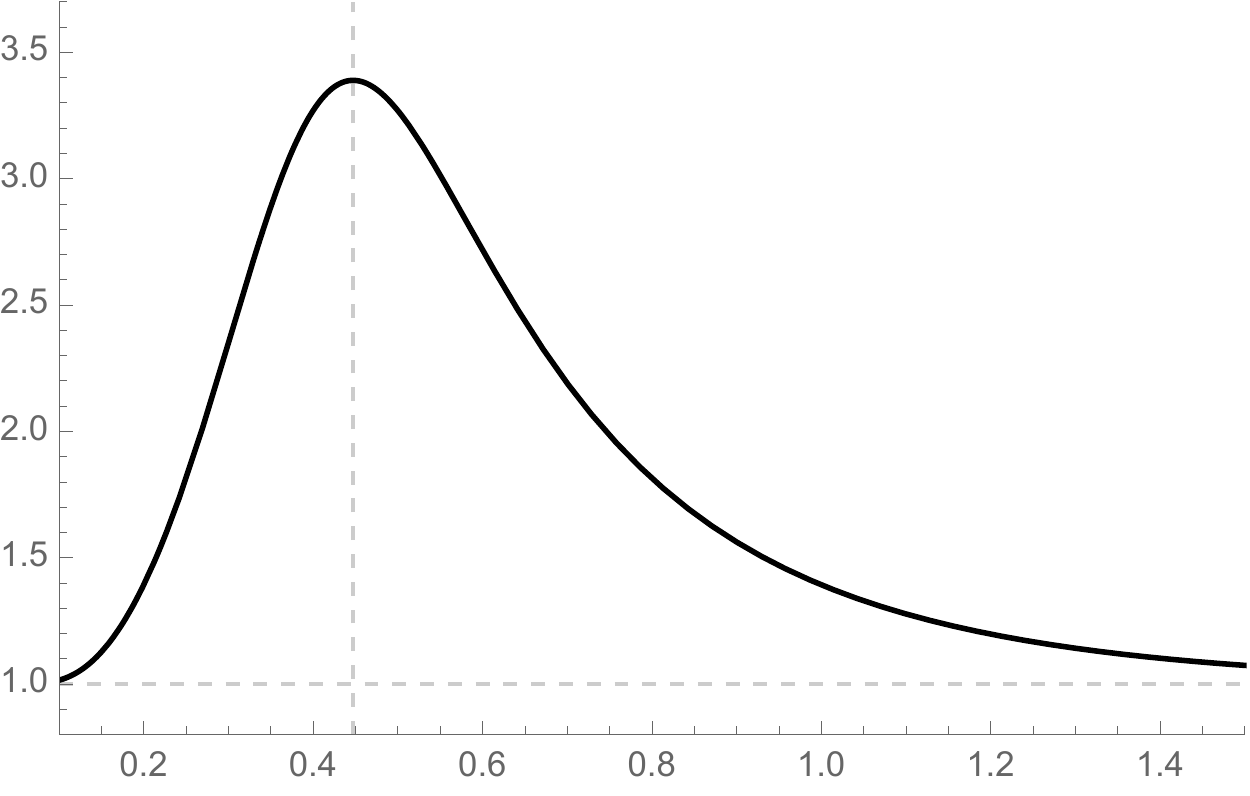}
\end{center}
\caption{\label{figure-k12h4} The modified secrecy functions $\widetilde\Xi_{K_{12}}$ and $\widetilde\Xi_{H4}$ of Examples \ref{K12example} and \ref{H4example}.}
\end{figure}

The proof of Theorem \ref{poly} also works in the case $\ell=2$ which has been treated in \cite{e-hsethuraman}.
In the case that $\ell$ is a composite number, we can prove the following theorem which is weaker than the theorem above.

\begin{theorem}\label{polycomposite} Let $\ell\in \{6,14,15\}$, let $\Lambda$ be a strongly $\ell$-modular lattice, which is rationally equivalent to $(C^{\ell})^k$, where $k\in\mathbb Z_+$, and let
\[
\Theta_{\Lambda}=\Theta_{C^\ell}^k\sum_{i=0}^{N}c_i\,g_\ell^i,
\]
where $N\in\mathbb Z_+$ and the coefficients $c_0$, $c_1$, \dots, $c_N$ are real numbers. Assume that
\[
\sum_{i=0}^{N}c_i\,x^i\geqslant\sum_{i=0}^N c_i\,g_\ell^i\!\left(\frac1{\sqrt\ell}\right)
\]
for all $x\in\left]0,g_\ell(1/\smash{\sqrt\ell})\right]$.
Then the modified secrecy function $\widetilde\Xi_\Lambda$ has a unique global maximum at the natural symmetry point $1/\sqrt\ell$.
\end{theorem}

\begin{remark}This theorem has only one direction: even if the polynomial does not attain its minimum in the interval at the given point, it might happen that the lattice satisfies the conjecture.\end{remark}

\begin{proof}[Proof of Theorem \ref{polycomposite}]
We start by factoring the modified secrecy function into two parts
\[\widetilde\Xi_\Lambda
=\frac{\Theta_{D^\ell}^{n/2}}{\Theta_\Lambda}
=\frac{\Theta_{D^\ell}^{n/2}}{\Theta_{C^\ell}^k}\cdot
\frac{\Theta_{C^\ell}^k}{\Theta_\Lambda}.\]
The first factor $\Theta_{D^\ell}^{n/2}/\Theta_{C^\ell}^k$ has a unique global maximum at the point $1/\sqrt\ell$ by Theorem \ref{modified-theta-ell-max-general}, and so it is enough to prove that the second quotient has a global maximum at the point $1/\sqrt\ell$. But this factor is
\[
\frac{\Theta_{C^{\ell}}^{k}}{\Theta_{\Lambda}}
%=\frac{\Theta_{C^{\ell}}^{k}(y)}{g_1^k(y)\sum\limits_{i=0}^{\lfloor k\,\mathrm{ord}_1(g_1)\rfloor}c_i\,g_2^i(y)}
=\left(\sum_{i=0}^{N}c_i\,g_\ell^i\right)^{-1}.
\]
Hence, it suffices to consider the expression $\sum_{i=0}^{N}c_i\,g_\ell^i$. Let us now analyse the behaviour of the function $g_\ell$. By Theorem \ref{eta-quotient-theorem} below the function $g_\ell$ is strictly increasing in $\left]0,1/\smash{\sqrt\ell}\right]$, and strictly decreasing in $\left[1/\smash{\sqrt\ell},\infty\right[$, with $g_\ell(y)\longrightarrow0+$ as $y\longrightarrow0+$ or $y\longrightarrow\infty$. Thus, the conjecture holds for $\Lambda$, provided that the polynomial $\sum_{i=0}^{N}c_i\,x^i$ has in the interval $\left]0,g_\ell(1/\smash{\sqrt\ell})\right]$ a global minimum at the point $g_\ell(1/\sqrt\ell)$, and we are done.
\end{proof}

\begin{figure}[h]
\begin{center}
\begin{tabular}{|c|c|}
\hline
$\ell$&$g_\ell(1/\sqrt\ell)\vphantom{\Big|}$\\\hline
3&0.0625000\\
5&0.0954915\\
6&0.133975\\
7&0.125000\\
11&0.176101\\
14&0.228788\\
15&0.250000\\
23&0.284920\\\hline
\end{tabular}
\end{center}
\caption{\label{g-ell-table} Numerical values for the expressions $g_\ell(1/\sqrt\ell)$ appearing in Theorems \ref{poly} and \ref{polycomposite}.}
\end{figure}

The following theorem is crucial for understanding the behaviour of $g_\ell$ in the above theorems. It depends on a technical lemma on some finer properties of the function $\log\eta(\exp)$. The lemma is stated and proved in Section \ref{technical-eta-section} below.

\begin{figure}[h]
\begin{center}
\includegraphics[scale=.6]{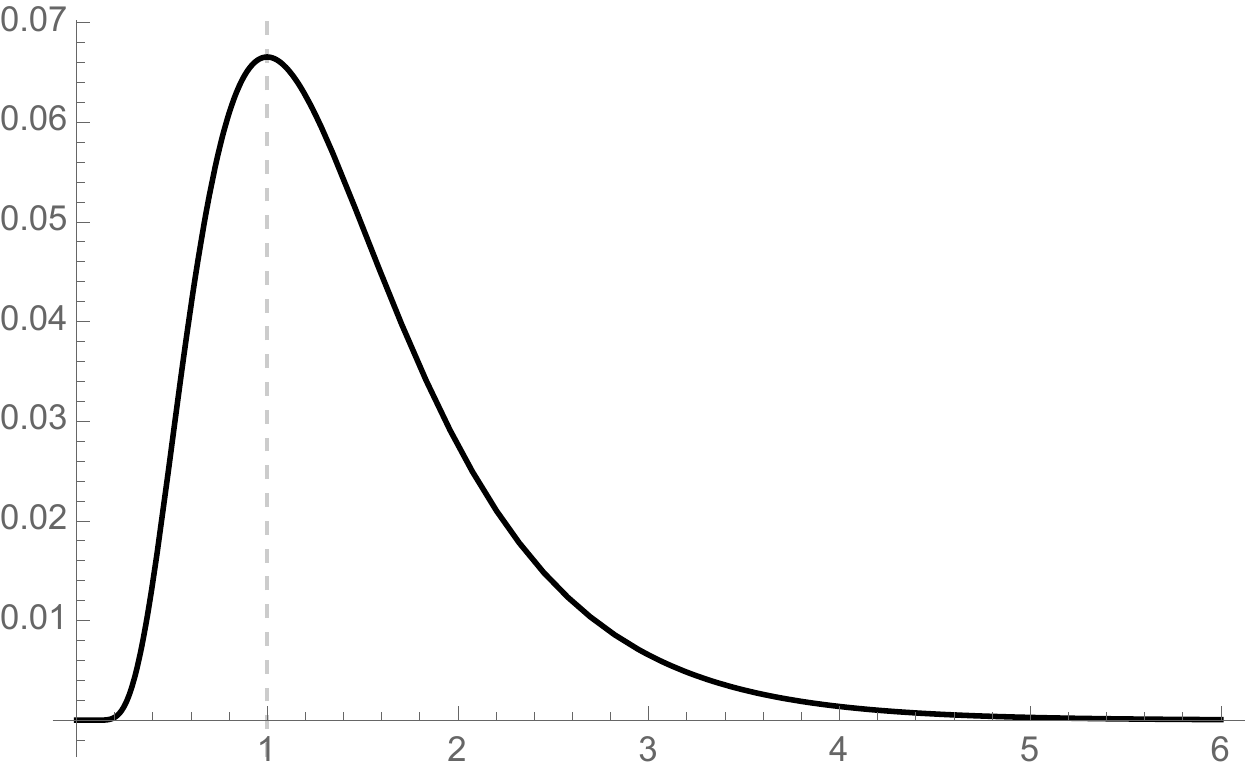}
\end{center}
\caption{\label{figure-eta-quotient} The function $g(y)$ of Theorem \ref{eta-quotient-theorem} with $\kappa=2$, $\lambda=5$ and $\ell=3$.}
\end{figure}

\begin{theorem}\label{eta-quotient-theorem}
Let $\kappa,\lambda\in\mathbb R_+$ with $1\leqslant\kappa<\lambda$, and let $\ell\in\mathbb Z_+$. Let us define a function $g\colon\mathbb R_+\longrightarrow\mathbb R_+$ by setting
\[g(y)=\frac{\eta^{(\ell)}\!\left(\displaystyle{\frac{\lambda y}{\sqrt\ell}}\right)\eta^{(\ell)}\!\left(\displaystyle{\frac y{\lambda\sqrt\ell}}\right)}{\eta^{(\ell)}\!\left(\displaystyle{\frac{\kappa y}{\sqrt\ell}}\right)\eta^{(\ell)}\!\left(\displaystyle{\frac y{\kappa\sqrt\ell}}\right)}\]
for all $y\in\mathbb R_+$.

This function $g$ is a real-analytic function satisfying the modularity relation
\[g\!\left(\frac1y\right)=g(y)\]
for all $y\in\mathbb R_+$. Furthermore, it is strictly increasing in $\left]0,1\right]$, strictly decreasing in $\left[1,\infty\right[$, has a strict global maximum at the point $1$, its image set is $\left]0,g(1)\right]$, and furthermore, $g(y)\longrightarrow0$ as $y\longrightarrow\infty$ or $y\longrightarrow0+$.
\end{theorem}

\begin{proof}
The real-analyticity of $g$ follows from the real-analyticity of all the functions involved in its definition. The modularity relation is an immediate consequence of the modularity relation
\[\eta\!\left(\frac1y\right)=\sqrt y\,\eta(y),\]
which holds for all $y\in\mathbb R_+$. The conclusion about the image set will follow from the other remaining claims and the continuity of $g$. Of the limits of $g(y)$ as $y\longrightarrow0+$ or $y\longrightarrow\infty$ only one needs to be considered as the other case follows from the modularity relation. The limit $y\longrightarrow\infty$ is easily dealt with as $\eta(y)\sim e^{-\pi y/12}$ as $y\longrightarrow\infty$ and so $\eta^{(\ell)}(y)\sim e^{-\pi\sigma(\ell)y/12}$ as $y\longrightarrow\infty$, where $\sigma(N)$ denotes the sum of the positive divisors of $N$, and consequently,
\[g(y)\sim\exp\!\left(-\frac{\pi\,\sigma(\ell)}{12\,\sqrt\ell}\left(\lambda+\frac1\lambda-\kappa-\frac1\kappa\right)y\right)\]
as $y\longrightarrow\infty$, and the conclusion $g(y)\longrightarrow0$ as $y\longrightarrow\infty$ follows from the simple observation that $\lambda+1/\lambda>\kappa+1/\kappa$.

Finally, the claims about the global maximum at $1$ follows from the claimed monotonicity properties in $\left]0,1\right]$ and $\left[1,\infty\right[$, and by the modularity relation, it only remains to prove that $g(y)$ is strictly decreasing for $y\in\left]1,\infty\right[$. It is enough to prove that $\log g(e^x)$ is strictly decreasing for $x\in\mathbb R_+$. In terms of the function $f(x)$, given by $\log\eta(e^x)$ for $x\in\mathbb R$, and the parameters $h=\log\lambda$ and $k=\log\kappa$, this logarithm is, applying Lemma \ref{convolution-identity},
\begin{align*}
\log g(e^x)
&=\sum_{d\mid\ell}\left(f(x+\log d-\log\sqrt\ell+h)
-f(x+\log d-\log\sqrt\ell+k)\right.\\
&\qquad\qquad\left.-f(x+\log d-\log\sqrt\ell-k)
+f(x+\log d-\log\sqrt\ell-h)\right)\\
&=\sum_{d\mid\ell}\int\limits_{-\infty}^\infty f''(u+\log d-\log\sqrt\ell)\,T(u-x;k,h)\,\mathrm du
=\int\limits_{-\infty}^\infty f''(u)\,K(u-x)\,\mathrm du,
\end{align*}
where the kernel $K$ is given by
\[K(u)=\sum_{d\mid\ell}T(u-\log d+\log\sqrt\ell;k,h)\]
for $u\in\mathbb R$. It is easily seen that this is an even continuous compactly supported function taking only nonnegative values.

Since $f''$ is real-analytic and $K$ is continuous and compactly supported, the last integral is differentiable with derivative
\[\frac{\mathrm d}{\mathrm dx}\,\log g(e^x)=\int\limits_{-\infty}^\infty f'''(u)\,K(u-x)\,\mathrm du.\]
We only need to prove that this is strictly negative for $x\in\mathbb R_+$. But since $f'''$ is a strictly decreasing odd function in $\mathbb R$ by Theorem \ref{subtle-properties-of-eta}, and $K$ takes only nonnegative values and is even and not identically zero, we may estimate simply
\[\int\limits_{-\infty}^\infty f'''(u)\,K(u-x)\,\mathrm du
<\int\limits_{-\infty}^\infty f'''(u)\,K(u)\,\mathrm du=0,\]
and we are done.
\end{proof}

\section{Technical lemmas on $\vartheta_3$}\label{technical-theta-section}

We now state and prove the technical lemmas that are needed to prove Theorem \ref{theta-quotients}. It turns out that the following lemma has appeared as Proposition 5.14 in Faulhuber's dissertation \cite{Faulhuber}, but this section is nonetheless self-contained.

\begin{lemma}\label{theta-min-technical-lemma}
The second derivative of the function $\log\vartheta_3(e^x)$, defined for $x\in\mathbb R$, is an even function which is strictly decreasing for $x\in\left[0,\infty\right[$, and strictly increasing for $x\in\left]-\infty,0\right]$.
\end{lemma}

\begin{proof}
Let us simplify notation by writing $f(x)=\log\vartheta_3(e^x)$ for $x\in\mathbb R$.
By the modularity relation of $\vartheta_3$ we have
\[\vartheta_3\!\left(\frac1y\right)=\sqrt y\,\vartheta_3(y),\]
for $y\in\mathbb R_+$,
so that
\[\log\vartheta_3(e^{-x})=\frac x2+\log\vartheta_3(e^x),\]
for $x\in\mathbb R$, so that the second derivative $f''$ is an even function of $x$, and it is enough to prove that it is strictly decreasing for $x\in\left[0,\infty\right[$.

The second derivative is
\[f''(x)
=\frac{e^x\,\vartheta_3'(e^x)\,\vartheta_3(e^x)+e^{2x}\,\vartheta_3''(e^x)\,\vartheta_3(e^x)-e^{2x}\left(\vartheta_3'(e^x)\right)^2}{\vartheta_3^2(e^x)}.\]

We will prove the claim first for $x\in\left[0,\log(3/2)\right]$. Since $f''(-x)=f''(x)$ for all $x\in\mathbb R$, it is enough to prove that the function $f''$ is strictly concave in $\left[0,\log(3/2)\right]$, or equivalently, that the fourth derivative $f''''(x)$ is strictly negative for $x\in\left[0,\log(3/2)\right]$, but this is established in Lemma \ref{technical-lemma-on-fourth-derivative} below.

Thus it remains to prove that the function
\[\frac{y\,\vartheta_3'(y)\,\vartheta_3(y)+y^2\,\vartheta_3''(y)\,\vartheta_3(y)-y^2\left(\vartheta_3'(y)\right)^2}{\vartheta_3^2(y)}\]
is strictly decreasing for $y\in\left[3/2,\infty\right[$. To do so, we will prove that its derivative, given by the expression
\begin{multline*}
\frac{\vartheta_3'(y)\,\vartheta_3^2(y)+3\,y\,\vartheta_3''(y)\,\vartheta_3^2(y)-3\,y\left(\vartheta_3'(y)\right)^2\vartheta_3(y)}{\vartheta_3^3(y)}\\
+\frac{y^2\,\vartheta_3'''(y)\,\vartheta_3^2(y)-3\,y^2\,\vartheta_3''(y)\,\vartheta_3'(y)\,\vartheta_3(y)+2\,y^2\left(\vartheta_3'(y)\right)^3}{\vartheta_3^3(y)},
\end{multline*}
takes only strictly negative values for $y\in\left[3/2,\infty\right[$. Since the denominator takes only strictly positive values, we may focus solely on the numerator. But the strict negativity of the numerator for $y\in\left[3/2,\infty\right[$ is shown in Lemma \ref{technical-lemma-on-third-derivative} below.
\end{proof}

\begin{lemma}\label{technical-lemma-on-fourth-derivative}
Let $f(x)=\log\vartheta_3(e^x)$ for $x\in\mathbb R$. Then
\[f''''(x)<0\]
for $x\in\left[0,\log(3/2)\right]$.
\end{lemma}

\begin{proof}
Let us write $g(y)=\log\vartheta_3(y)$ for $y\in\mathbb R_+$. The first derivative of this expression is
\[\frac{\mathrm d}{\mathrm dy}\,g(y)=\frac{\vartheta_3'(y)}{\vartheta_3(y)},\]
its second derivative is
\[\frac{\mathrm d^2}{\mathrm dy^2}\,g(y)=\frac{\vartheta_3''(y)}{\vartheta_3(y)}-\frac{\left(\vartheta_3'(y)\right)^2}{\vartheta_3^2(y)},\]
its third derivative is
\[\frac{\mathrm d^3}{\mathrm dy^3}\,g(y)=\frac{\vartheta_3'''(y)}{\vartheta_3(y)}-\frac{3\,\vartheta_3''(y)\,\vartheta_3'(y)}{\vartheta_3^2(y)}+\frac{2\left(\vartheta_3'(y)\right)^3}{\vartheta_3^3(y)},\]
and its fourth derivative is
\[\frac{\mathrm d^4}{\mathrm dy^4}\,g(y)=\frac{\vartheta_3''''(y)}{\vartheta_3(y)}-\frac{4\,\vartheta_3'''(y)\,\vartheta_3'(y)}{\vartheta_3^2(y)}-\frac{3\left(\vartheta_3''(y)\right)^2}{\vartheta_3^2(y)}+\frac{12\,\vartheta_3''(y)\left(\vartheta_3'(y)\right)^2}{\vartheta_3^3(y)}-\frac{6\left(\vartheta_3'(y)\right)^4}{\vartheta_3^4(y)},\]
and the fourth derivative $f''''$ is
\[\frac{\mathrm d^4}{\mathrm dx^4}\,f(x)=
\frac{\mathrm d^4}{\mathrm dx^4}\,g(e^x)
=g''''(e^x)\,e^{4x}+6\,g'''(e^x)\,e^{3x}+7\,g''(e^x)\,e^{2x}+g'(e^x)\,e^x.\]
Thus, we need to prove that the expression
\[h(y)=y^4\,g''''(y)+6\,y^3\,g'''(y)+7\,y^2\,g''(y)+y\,g'(y)\]
is strictly negative for $y\in\left[1,3/2\right]$.

Let $m,M\in\left[1,\infty\right[$ with $m<M$. Then, for $y\in\left[m,M\right]$, we may use Lemma \ref{theta3-asymptotics-with-notation} to rewrite and estimate the expression as follows:
\begin{align*}
h(y)&=y^4\left(\frac{\vartheta_3''''(y)}{\vartheta_3(y)}-\frac{4\,\vartheta_3'''(y)\,\vartheta_3'(y)}{\vartheta_3^2(y)}-\frac{3\left(\vartheta_3''(y)\right)^2}{\vartheta_3^2(y)}+\frac{12\,\vartheta_3''(y)\left(\vartheta_3'(y)\right)^2}{\vartheta_3^3(y)}-\frac{6\left(\vartheta_3'(y)\right)^4}{\vartheta_3^4(y)}
\right)\\
&\qquad+6\,y^3\left(\frac{\vartheta_3'''(y)}{\vartheta_3(y)}-\frac{3\,\vartheta_3''(y)\,\vartheta_3'(y)}{\vartheta_3^2(y)}+\frac{2\left(\vartheta_3'(y)\right)^3}{\vartheta_3^3(y)}
\right)\\
&\qquad+7\,y^2\left(\frac{\vartheta_3''(y)}{\vartheta_3(y)}-\frac{\left(\vartheta_3'(y)\right)^2}{\vartheta_3^2(y)}\right)
+y\,\frac{\vartheta_3'(y)}{\vartheta_3(y)}\\
&<
\left(\frac{M^4\,\Theta_{3,4}(m)}{\vartheta_{3,0}(M)}
-\frac{4\,m^4\,\vartheta_{3,3}(M)\,\vartheta_{3,1}(M)}{\Theta_{3,0}^2(m)}
-\frac{3\,m^4\left(\vartheta_{3,2}(M)\right)^2}{\Theta_{3,0}^2(m)}\right.\\
&\qquad\qquad\left.
+\frac{12\,M^4\,\Theta_{3,2}(m)\left(\Theta_{3,1}(m)\right)^2}{\vartheta_{3,0}^3(M)}
-\frac{6\,m^4\left(\vartheta_{3,1}(M)\right)^4}{\Theta_{3,0}^4(m)}
\right)\\
&\qquad+6\left(-\frac{m^3\,\vartheta_{3,3}(M)}{\Theta_{3,0}(m)}+\frac{3\,M^3\,\Theta_{3,2}(m)\,\Theta_{3,1}(m)}{\vartheta_{3,0}^2(M)}-\frac{2\,m^3\left(\vartheta_{3,1}(M)\right)^3}{\Theta_{3,0}^3(m)}
\right)\\
&\qquad+7\left(\frac{M^2\,\Theta_{3,2}(m)}{\vartheta_{3,0}(M)}-\frac{m^2\,\left(\vartheta_{3,1}(M)\right)^2}{\Theta_{3,0}^2(m)}\right)
-\frac{m\,\vartheta_{3,1}(M)}{\Theta_{3,0}(m)}.
\end{align*}
Using this upper bound, it is easy to check numerically, that $h(y)<-0.{16}$ for all
\[y\in\left[1+\frac{k-1}{1000},1+\frac k{1000}\right]\]
for each $k\in\left\{1,2,\ldots,500\right\}$ separately.
\end{proof}

\begin{lemma}\label{technical-lemma-on-third-derivative}
Let $y\in\left[3/2,\infty\right[$. Then
\begin{multline*}
\vartheta_3'(y)\,\vartheta_3^2(y)+3\,y\,\vartheta_3''(y)\,\vartheta_3^2(y)-3\,y\left(\vartheta_3'(y)\right)^2\vartheta_3(y)\\
+y^2\,\vartheta_3'''(y)\,\vartheta_3^2(y)-3\,y^2\,\vartheta_3''(y)\,\vartheta_3'(y)\,\vartheta_3(y)+2\,y^2\left(\vartheta_3'(y)\right)^3<0.
\end{multline*}
\end{lemma}

\begin{proof}
Using Lemma \ref{theta3-asymptotics-with-notation} below, the left-hand side is
\begin{multline*}
<-\vartheta_{3,1}(y)\,\vartheta_{3,0}^2(y)+3\,y\,\Theta_{3,2}(y)\,\Theta_{3,0}^2(y)-3\,y\left(\vartheta_{3,1}(y)\right)^2\vartheta_{3,0}(y)\\
-y^2\,\vartheta_{3,3}(y)\,\vartheta_{3,0}^2(y)+3\,y^2\,\Theta_{3,2}(y)\,\Theta_{3,1}(y)\,\Theta_{3,0}(y)-2\,y^2\left(\vartheta_{3,1}(y)\right)^3.
\end{multline*}
We may further absorb the terms involving $e^{-9\pi y}$ into those involving $e^{-4\pi y}$, remembering that $e^{-\pi y}<1/100$ for $y\geqslant3/2$, leading to the upper bound
\begin{align*}
<&-\left(2\pi\,e^{-\pi y}+8\pi\,e^{-4\pi y}\right)\left(1+2\,e^{-\pi y}+2\,e^{-4\pi y}\right)^2\\
&+3\,y\left(2\,\pi^2\,e^{-\pi y}+33\,\pi^2\,e^{-4\pi y}\right)\left(1+2\,e^{-\pi y}+3\,e^{-4\pi y}\right)^2\\
&-3\,y\left(2\pi\,e^{-\pi y}+8\pi\,e^{-4\pi y}\right)^2\left(1+2\,e^{-\pi y}+2\,e^{-4\pi y}\right)\\
&-y^2\left(2\,\pi^3\,e^{-\pi y}+128\,\pi^3\,e^{-4\pi y}\right)\left(1+2\,e^{-\pi y}+2\,e^{-4\pi y}\right)^2\\
&+3\,y^2\left(2\,\pi^2\,e^{-\pi y}+33\,\pi^2\,e^{-4\pi y}\right)\left(2\pi\,e^{-\pi y}+9\pi\,e^{-4\pi y}\right)\left(1+2\,e^{-\pi y}+3\,e^{-4\pi y}\right)\\
&-2\,y^2\left(2\pi\,e^{-\pi y}+8\pi\,e^{-4\pi y}\right)^3.
\end{align*}
The last expression turns out to be
\begin{align*}
=&-\left( 2 \,\pi^3 \,y^2 -6 \,\pi^2 \,y +2 \pi\right)e^{-\pi y} 
+\left( 4 \,\pi^3\, y^2+ 12\, \pi^2 \,y-8 \pi \right)e^{-2 \pi y}-8 \pi\,e^{-3 \pi y}\\
&- \left(128\, \pi^3 \,y^2 - 99 \,\pi^2\, y +8 \pi \right)e^{-4 \pi y}
-\left( 268\, \pi^3 \,y^2 - 336 \,\pi^2\, y +40 \,\pi\right)e^{-5 \pi y} \\
&- \left( 180\, \pi^3\, y^2 - 252\, \pi^2\, y +48 \,\pi\right)e^{-6 \pi y}
+ \left( 379 \,\pi^3 \,y^2 +402\, \pi^2\, y-32 \,\pi \right)e^{-8 \pi y}\\
&+\left( 738 \,\pi^3 \,y^2 + 666\, \pi^2 \,y -72\, \pi\right)e^{-9 \pi y}
+\left( 1137 \,\pi^3\, y^2 + 507 \,\pi^2\, y-32 \,\pi \right)e^{-12 \pi y}.
\end{align*}
Here the contribution from the terms involving $e^{-3\pi y}$, $e^{-4\pi y}$, $e^{-5\pi y}$ and $e^{-6\pi y}$ are clearly strictly negative for $y\in\left[3/2,\infty\right[$. Similarly, we may forget those terms not involving a power of $y$, since all of them are strictly negative as well. Thus, we may continue our estimations
\begin{align*}
<&-\left( 2 \,\pi^3 \,y^2 -6 \,\pi^2 \,y\right)e^{-\pi y} 
+\left( 4 \,\pi^3\, y^2+ 12\, \pi^2 \,y \right)e^{-2 \pi y}\\
&+ \left( 379 \,\pi^3 \,y^2 +402\, \pi^2\, y \right)e^{-8 \pi y}
+\left( 738 \,\pi^3 \,y^2 + 666\, \pi^2 \,y \right)e^{-9 \pi y}\\
&+\left( 1137 \,\pi^3\, y^2 + 507 \,\pi^2\, y \right)e^{-12 \pi y}.
\end{align*}
Now we extract a common factor $-y^2\,e^{-\pi y}$, and use the lower bound $y\geqslant3/2$ to continue
\begin{multline*}
\leqslant y^2\,e^{-\pi y}\left(-2\,\pi^3+4\,\pi^2+4\,\pi^3\,e^{-3\pi/2}+8\,\pi^2\,e^{-3\pi/2}\right.\\
\left.+379\,\pi^3\,e^{-21\pi/2}+268\,\pi^2\,e^{-21\pi/2}+738\,\pi^3\,e^{-24\pi/2}+444\,\pi^2\,e^{-24\pi/2}\right.\\
\left.+1137\,\pi^3\,e^{-33\pi/2}+338\,\pi^2\,e^{-33\pi/2}\right).
\end{multline*}
Finally, this last expression is easily seen to be $<-20\,y^2\,e^{-\pi y}$, and we are done.
\end{proof}

\begin{lemma}\label{theta3-asymptotics-with-notation}
Let $y\in\left[1,\infty\right[$. Then the $\vartheta_3$-function and its derivatives satisfy
\[0<\vartheta_{3,\nu}(y)<(-1)^\nu\,\vartheta_3^{(\nu)}(y)<\Theta_{3,\nu}(y)\]
for all $y\in\left[1,\infty\right[$, for each $\nu\in\left\{0,1,2,3,4\right\}$, where
\[\begin{cases}
\vartheta_{3,0}(y)=1+2\,e^{-\pi y}+2\,e^{-4\pi y}+2\,e^{-9\pi y},\\
\vartheta_{3,1}(y)=2\pi\,e^{-\pi y}+8\pi\,e^{-4\pi y}+18\pi\,e^{-9\pi y},\\
\vartheta_{3,2}(y)=2\,\pi^2\,e^{-\pi y}+32\,\pi^2\,e^{-4\pi y}+162\,\pi^2\,e^{-9\pi y},\\
\vartheta_{3,3}(y)=2\,\pi^3\,e^{-\pi y}+128\,\pi^3\,e^{-4\pi y}+1458\,\pi^3\,e^{-9\pi y},\\
\vartheta_{3,4}(y)=2\,\pi ^4\,e^{-\pi y}+512\,\pi^4\,e^{-4\pi y}+13122\,\pi^4\,e^{-9\pi y},
\end{cases}\]
and
\[\begin{cases}
\Theta_{3,0}(y)=1+2\,e^{-\pi y}+2\,e^{-4\pi y}+3\,e^{-9\pi y},\\
\Theta_{3,1}(y)=2\pi\,e^{-\pi y}+8\pi\,e^{-4\pi y}+19\pi\,e^{-9\pi y},\\
\Theta_{3,2}(y)=2\,\pi^2\,e^{-\pi y}+32\,\pi^2\,e^{-4\pi y}+163\,\pi^2\,e^{-9\pi y},\\
\Theta_{3,3}(y)=2\,\pi^3\,e^{-\pi y}+128\,\pi^3\,e^{-4\pi y}+1459\,\pi^3\,e^{-9\pi y},\\
\Theta_{3,4}(y)=2\,\pi ^4\,e^{-\pi y}+512\,\pi^4\,e^{-4\pi y}+13123\,\pi^4\,e^{-9\pi y}.
\end{cases}\]
\end{lemma}

\begin{proof}
Let $y\in\left[1,\infty\right[$ and $\nu\in\left\{0,1,2,3,4\right\}$. The lower bounds $0<\vartheta_{3,\nu}(y)<(-1)^\nu\,\vartheta_3^{(\nu)}(y)$ hold trivially as $\vartheta_{3,\nu}(y)$ are just the beginning of the Fourier series representation of $(-1)^\nu\,\vartheta_3^{(\nu)}(y)$. Thus, it is enough to prove the upper bounds involving $\Theta_{3,\nu}(y)$. This is achieved by estimating
\[
0<(-1)^\nu\,\vartheta_3^{(\nu)}(y)-\vartheta_{3,\nu}(y)
=2\,\pi^\nu\sum_{n=4}^\infty n^{2\nu}\,e^{-\pi n^2y}
<2\,\pi^\nu\sum_{n=16}^\infty n^\nu\,e^{-\pi ny}
<2\,\pi^\nu\int\limits_{15}^\infty t^\nu\,e^{-\pi ty}\,\mathrm dt,
\]
where we apply the simple fact that the function $t^\nu\,e^{-\pi ty}$ is strictly decreasing in $t$ for $t\in\left[15,\infty\right[$ for each fixed $y\in\left[1,\infty\right[$ and for each $\nu\in\left\{0,1,2,3,4\right\}$.

The proof is finished by showing that the expression $2\,\pi^\nu\int_{15}^\infty\ldots$ is smaller than $\pi^\nu\,e^{-9\pi y}$. In the case $\nu=4$ we get
\begin{multline*}
2\,\pi^4\int\limits_{15}^\infty t^4\,e^{-\pi ty}\,\mathrm dt
=\left(\frac{101250}{\pi y}+\frac{27000}{\pi^2\,y^2}+\frac{5400}{\pi^3\,y^3}+\frac{720}{\pi^4\,y^4}+\frac{48}{\pi^5\,y^5}\right)\pi^4\,e^{-15\pi y}\\
\leqslant\left(\frac{101250}\pi+\frac{27000}{\pi^2}+\frac{5400}{\pi^3}+\frac{720}{\pi^4}+\frac{48}{\pi^5}\right)e^{-6\pi}\,\pi^4\,e^{-9\pi y}
<\pi^4\,e^{-9\pi y}.
\end{multline*}
The other cases $\nu\in\left\{0,1,2,3\right\}$ are similar but slightly simpler.
\end{proof}

\section{A technical lemma on $\eta$}\label{technical-eta-section}

Proofs of the polynomization results in Section \ref{polynomization} are based on some good properties of the function $\log\eta(\exp)$ on the real line. The crucial features of this function are listed in the following theorem.

\begin{theorem}\label{subtle-properties-of-eta}
Let $f\colon\mathbb R\longrightarrow\mathbb R$ be defined by $f(x)=\log\eta(e^x)$ for all $x\in\mathbb R$. The function $f$ is a real-analytic strictly concave function, and the second derivative $f''$ is a strictly concave even function. Furthermore, the third derivative $f'''$ is a strictly decreasing odd function.
\end{theorem}

\begin{proof}
It is clear that $f$ is real-analytic. Taking logarithms in the modularity relation of $\eta$ gives
\[f(-x)=\frac x2+f(x)\]
for all $x\in\mathbb R$. Differentiating two, three and four times gives
\[f''(-x)=f''(x),\qquad f'''(-x)=-f'''(x),\qquad\text{and}\qquad f''''(-x)=f''''(x),\]
respectively, for all $x\in\mathbb R$. Furthermore, the desired properties of $f'''$ follow immediately from the desired properties of $f''$. Thus, it only remains to prove that $f$ and $f''$ are both strictly concave functions. Furthermore, since $f''$ and $f''''$ are even, it is enough to prove that $f''(x)<0$ and $f''''(x)<0$ for all $x\in\left[0,\infty\right[$.

We will start with the series representation
\[f(x)=-\frac{\pi e^x}{12}+\sum_{n=1}^\infty\log\left(1-e^{-2\pi ne^x}\right),\]
which converges absolutely for all $x\in\mathbb R$, and uniformly in any bounded interval of $\mathbb R$. We recall the Taylor expansion
\[\log\left(1-z\right)=-\sum_{m=1}^\infty\frac{z^m}m,\]
which holds for all $z\in\left]-1,1\right[$, and where the series on the right converges absolutely for all such $z$, and uniformly when $z$ is restricted to a closed subinterval of $\left]-1,1\right[$. Using this expansion we may continue by writing
\[f(x)=-\frac{\pi e^x}{12}-\sum_{n=1}^\infty\sum_{m=1}^\infty\frac1m\,e^{-2\pi mne^x}.\]
To prove that $f$ is srictly concave, it is enough to show that the expression $\exp(-2\pi mne^x)$ gives a strictly convex function for $x\in\left[0,\infty\right[$. The second derivative of the expression is
\begin{align*}
\frac{\mathrm d^2}{\mathrm dx^2}\left(e^{-2\pi mne^x}\right)
%&=\frac{\mathrm d}{\mathrm dx}\left(-2\pi m n\,e^x\,e^{-2\pi mne^x}\right)\\
%&=-2\pi m n\,e^x\,e^{-2\pi mne^x}+4\,\pi^2\,m^2\,n^2\,e^{2x}\,e^{-2\pi mne^x}\\
&=2\pi m n\,e^x\,e^{-2\pi mne^x}\left(-1+2\pi mn\,e^x\right),
\end{align*}
and this is strictly positive since $2\pi mn\,e^{x}\geqslant2\pi>1$ for all $x\in\left[0,\infty\right[$, $m\in\mathbb Z_+$ and $n\in\mathbb Z_+$.

To prove that $f''''$ is strictly concave, it is enough to prove that the above expression $\exp(-2\pi mne^x)$ has a strictly positive fourth derivative for $x\in\left[0,\infty\right[$. Its %third derivative is
%\begin{align*}
%\frac{\mathrm d^3}{\mathrm dx^3}\left(e^{-2\pi mne^x}\right)
%&=-2\pi mn\,e^x\,e^{-2\pi mne^x}+12\,\pi^2\,m^2\,n^2\,e^{2x}\,e^{-2\pi mne^x}\\
%&\qquad-8\,\pi^3\,m^3\,n^3\,e^{3x}\,e^{-2\pi mne^x},
%\end{align*}
%and the
fourth derivative is
\begin{align*}
&\frac{\mathrm d^4}{\mathrm dx^4}\left(e^{-2\pi mne^x}\right)%\\
%\quad=-2\pi mn\,e^x\,e^{-2\pi mne^x}+28\,\pi^2\,m^2\,n^2\,e^{2x}\,e^{-2\pi mne^x}\\
%&\qquad-48\,\pi^3\,m^3\,n^3\,e^{3x}\,e^{-2\pi mne^x}
%+16\,\pi^4\,m^4\,n^4\,e^{4x}\,e^{-2\pi mne^x}\\
%&\quad
=\left(-1+14\,\pi mn\,e^x-24\,\pi^2\,m^2\,n^2\,e^{2x}+8\,\pi^3\,m^3\,n^3\,e^{3x}\right)2\pi mn\,e^x\,e^{-2\pi mne^x}.
\end{align*}
But this last expression is easily seen to be strictly positive, since for all $x\in\left[0,\infty\right[$, $m\in\mathbb Z_+$ and $n\in\mathbb Z_+$ we may easily estimate
\[14\,\pi mn\,e^x\geqslant14\,\pi>1\]
as well as
\[8\,\pi^3\,m^3\,n^3\,e^{3x}>24\,\pi^2\,m^3\,n^3\,e^{3x}\geqslant24\,\pi^2\,m^2\,n^2\,e^{2x},\]
and we are done.
\end{proof}

\end{document}